\documentclass[a4paper,10pt]{article}

\title{Some Notes on Temporal Justification Logic}
\author{Samuel Bucheli}
\date{October 25, 2015}

\usepackage{amsmath,amssymb,amsthm}
\usepackage{colonequals}

\newcommand{\Prop}{\textsf{Prop}}
\newcommand{\Formulae}{\textsf{Fml}}

\newcommand{\lfalse}{\bot}
\newcommand{\ltrue}{\top}
\newcommand{\lneg}{\neg}

\newcommand{\propax}{\ensuremath{(\textsf{Prop})}}
\newcommand{\lrule}[2]{\displaystyle{\frac{#1}{#2}}}
\newcommand{\mprule}{\ensuremath{(\textsf{MP})}}
\newcommand{\limplies}{\rightarrow}
\newcommand{\liff}{\leftrightarrow}

\newcommand{\lnext}{\bigcirc}
\newcommand{\lalways}{\Box}
\newcommand{\leventually}{\Diamond}
\newcommand{\luntil}{{\,\mathcal{U}\,}}

\newcommand{\kax}{\ensuremath{\textsf{-k}}}
\newcommand{\nextkax}{\ensuremath{(\lnext\kax)}}
\newcommand{\alwayskax}{\ensuremath{(\lalways\kax)}}
\newcommand{\funax}{\ensuremath{(\textsf{fun})}}
\newcommand{\mixax}{\ensuremath{(\textsf{mix})}}
\newcommand{\indax}{\ensuremath{(\textsf{ind})}}
\newcommand{\uoneax}{\ensuremath{(\luntil\textsf{1})}}
\newcommand{\utwoax}{\ensuremath{(\luntil\textsf{2})}}
\newcommand{\necrule}{\ensuremath{\textsf{-nec}}}
\newcommand{\nextnecrule}{\ensuremath{(\lnext\necrule)}}
\newcommand{\alwaysnecrule}{\ensuremath{(\lalways\necrule)}}
\newcommand{\uindrule}{\ensuremath{(\luntil\textsf{-ind})}}

\newcommand{\LTL}{\textsf{LTL}}
\newcommand{\LTLalt}{\LTL^\text{alt}}
\newcommand{\JFiveLTL}{\textsf{J5LTL}}

\newcommand{\lknows}{\mathsf{K}}
\newcommand{\lconsiders}{\mathsf{L}}

\newcommand{\SFour}{\textsf{S4}}
\newcommand{\SFive}{\textsf{S5}}

\newcommand{\nlax}{\ensuremath{(\textsf{nl})}}
\newcommand{\prax}{\ensuremath{(\textsf{pr})}}
\newcommand{\nlsyncax}{\ensuremath{(\textsf{nlsync})}}
\newcommand{\prsyncax}{\ensuremath{(\textsf{prsync})}}
\newcommand{\notquiteprax}{\ensuremath{(\textsf{notquitepr})}}
\newcommand{\knowexchax}{\ensuremath{(\textsf{knowexch})}}

\newcommand{\CTerms}{\textsf{Const}}
\newcommand{\VTerms}{\textsf{Var}}
\newcommand{\Terms}{\textsf{Tm}}

\newcommand{\jbox}[1]{\left[#1\right]\!}
\newcommand{\tapp}{\cdot}
\newcommand{\tsum}{+}
\newcommand{\tinspect}{!}
\newcommand{\tneginspect}{?}

\newcommand{\tnext}{\Rrightarrow}
\newcommand{\tprev}{\Lleftarrow}
\newcommand{\talwaysaccess}{\Downarrow}
\newcommand{\tgeneralize}{\Uparrow}
\newcommand{\tnextaccess}{\downarrow}

\newcommand{\JFive}{\textsf{J5}}

\newcommand{\appax}{\ensuremath{(\textsf{application})}}
\newcommand{\sumax}{\ensuremath{(\textsf{sum})}}
\newcommand{\posintax}{\ensuremath{(\textsf{positive introspection})}}
\newcommand{\negintax}{\ensuremath{(\textsf{negative introspection})}}
\newcommand{\refax}{\ensuremath{(\textsf{reflexivity})}}
\newcommand{\constnecrule}{\ensuremath{(\textsf{const}\necrule)}}

\newcommand{\CS}{\mathcal{CS}}

\newcommand{\numberofagents}{h}
\newcommand{\agent}{i}

\newcommand{\taind}{\textsf{aind}}
\newcommand{\tahead}{\textsf{ahead}}
\newcommand{\tatail}{\textsf{atail}}
\newcommand{\tuhead}{\textsf{uhead}}
\newcommand{\tutail}{\textsf{utail}}
\newcommand{\tuappend}{\textsf{uappend}}
\newcommand{\nextappax}{\ensuremath{(\lnext\textsf{-application})}}
\newcommand{\alwaysappax}{\ensuremath{(\lalways\textsf{-application})}}
\newcommand{\aindax}{\ensuremath{(\lalways\textsf{-ind})}}
\newcommand{\amixax}{\ensuremath{(\lalways\textsf{-mix})}}
\newcommand{\uax}{\ensuremath{(\luntil)}}
\newcommand{\uindax}{\ensuremath{(\luntil\textsf{-ind})}}
\newcommand{\umixax}{\ensuremath{(\luntil\textsf{-mix})}}

\newcommand{\alwaysaccessprinciple}{\ensuremath{(\lalways\textsf{-access})}}
\newcommand{\generalizeprinciple}{\ensuremath{(\textsf{generalize})}}
\newcommand{\nextaccessprinciple}{\ensuremath{(\lnext\textsf{-access})}}
\newcommand{\nextrightshiftprinciple}{\ensuremath{(\lnext\textsf{-right})}}
\newcommand{\nextleftshiftprinciple}{\ensuremath{(\lnext\textsf{-left})}}

\newcommand{\localstates}{\mathbf{L}}
\newcommand{\runs}{\mathcal{R}}
\newcommand{\system}{\mathcal{I}}
\newcommand{\accrel}{\sim}
\newcommand{\evidence}{\mathcal{E}}
\newcommand{\valuation}{\nu}
\newcommand{\entails}{\vDash}
\newcommand{\proves}{\vdash}
\newcommand{\localstatesequence}{\text{LSS}}
\newcommand{\futurelocalstatesequence}{\text{FLSS}}

\newcommand{\N}{\mathbb{N}}
\newcommand{\powerset}{\mathcal{P}}
\newcounter{enumsave}
\newtheorem{theorem}{Theorem}
\newtheorem{lemma}{Lemma}
\newtheorem{corollary}{Corollary}

\theoremstyle{definition}
\newtheorem{definition}{Definition}

\theoremstyle{remark}
\newtheorem*{remark}{Remark}
\newtheorem{question}{Question}

\renewcommand{\phi}{\varphi}

\begin{document}
 \maketitle
 
 \begin{abstract}
 Justification logics are modal-like logics with the additional capability of recording the reason, or justification, for modalities in syntactic structures, called justification terms.
 Justification logics can be seen as explicit counterparts to modal logic.
 The behavior and interaction of agents in distributed system is often modeled using logics of knowledge and time.
 In this paper, we sketch some preliminary ideas on how the modal knowledge part of such logics of knowledge and time could be replaced with an appropriate justification logic. 
 \end{abstract}
 
 \section{Introduction}
 \label{sect:Introduction}

Justification logics~\cite{ArtFit11SEP} are epistemic logics that feature explicit reasons for an agent's knowledge and belief. 
Originally, Artemov developed justification logic to provide a constructive semantics for intuitionistic logic. 
Later this type of  logics was introduced into formal epistemology, where it provides a novel approach to several epistemic puzzles and problems of multi-agent systems~\cite{Art06TCS,Art08RSL,Art10LNCS,BucKuzStu11JANCL,BucKuzStu11WoLLIC,KuzStu12AiML,ArtKuz14APAL,BucKuzStu14,KMOZS15}. 
Instead of an implicit statement~$\Box \phi$, which stands for  \emph{the agent knows~$\phi$}, justification logics include explicit statements of the form~$\jbox{t} \phi$, which mean \emph{$t$~justifies the agent's knowledge of~$\phi$}. 

 A common approach to model distributed systems of interacting agents is using logics of knowledge and time, with the interplay between these two modalities leading to interesting properties and questions~\cite{FHMV95,vdMW03,HvdMV04}.
 While knowledge in such systems has typically been modeled using the modal logic $\SFive$, it is a natural question to ask what happens when we model knowledge in such logics using a justification logic.
 
 In the following, we will sketch some preliminary ideas towards such a logic, and indicate further necessary work with appropriate questions.
 After briefly introducing the syntax in Section~\ref{sect:Syntax}, we propose an axiomatization in Section~\ref{sect:Axioms}, including possible additional principles. The resulting logic is illustrated with the proof of some simple properties in Section~\ref{sect:SomeProperties}. Finally, we introduce interpreted systems as the chosen semantics in Section~\ref{sect:Semantics} and we show soundness in Section~\ref{sect:Soundness}, where the question of completeness is also briefly addressed. The paper concludes with additional questions and remarks regarding further directions of research in Section~\ref{sect:Conclusion}.
 
 \section{Syntax}
 \label{sect:Syntax}
 
  In the following, let $\numberofagents$ be a fixed number of agents, $\CTerms$ a given set of proof constants, $\VTerms$ a given set of proof variables, and $\Prop$ a given set of atomic propositions..
 
 The set of justification terms $\Terms_\agent$ for agent~$1\leq\agent\leq\numberofagents$ is defined inductively by
 \[
  t^\agent \coloncolonequals c^\agent \mid x^\agent \mid \; \tinspect t^\agent \mid \; \tneginspect t^\agent \mid t^\agent \tsum t^\agent \mid t^\agent \tapp t^\agent \, ,
 \]
 where $c^\agent \in \CTerms_\agent$ and $x^\agent \in \VTerms_\agent$.
 
 The set of formulae $\Formulae$ is inductively defined by
 \[
  \phi \coloncolonequals P \mid \lfalse \mid \phi \limplies \phi \mid \lnext \phi \mid \lalways \phi \mid \phi \luntil \phi \mid \jbox{t^\agent}_\agent\phi \, ,
 \]
 where $1\leq\agent\leq\numberofagents$, $t^\agent \in \Terms_\agent$  and $P \in \Prop$.
 
We use the following usual abbreviations
 \begin{align*}
  \lneg \phi &\colonequals \phi \limplies \lfalse \, ,\\
  \ltrue &\colonequals \lneg \lfalse \, ,\\
  \phi \lor \psi &\colonequals \lneg \phi \limplies \psi \, ,\\
  \phi \land \psi &\colonequals \lneg (\lneg \phi \lor \lneg \psi \, ,\\
  \phi \liff \psi &\colonequals (\phi \limplies \psi) \land (\psi \limplies \phi) \, ,\\
  \leventually \phi &\colonequals \lneg \lalways \lneg \phi \, . \\
 \end{align*}
 
 Associativity and precedence of connectives, as well as the corresponding omission of brackets, are handled in the usual manner.
 
 \section{Axioms}
 \label{sect:Axioms}
  
  The axiom system for temporal justification logic consists of three parts, namely propositional logic, temporal logic, and justification logic.
  
 \subsection{Propositional Logic}
 For propositional logic, we take
 \begin{enumerate}
  \setcounter{enumi}{\theenumsave}
  \item all propositional tautologies \hfill \propax
  \setcounter{enumsave}{\theenumi}
 \end{enumerate}
 as axioms and the rule modus ponens, as usual:
 \[
   \lrule{\phi \quad \phi \limplies \psi}{\psi}\,\mprule \, .
 \]
 
 \subsection{Temporal Logic}
 For the temporal part, we use~\cite{Gor99} (see also Section~\ref{sect:AnAlternativePresentationOfTemporalLogic}), with axioms
 \begin{enumerate}
  \setcounter{enumi}{\theenumsave}
  \item $\lnext( \phi \limplies \psi) \limplies (\lnext \phi \limplies \lnext \psi)$ \hfill \nextkax
  \item $\lalways( \phi \limplies \psi) \limplies (\lalways \phi \limplies \lalways \psi)$ \hfill \alwayskax
  \item $\lnext \lneg \phi \liff \lneg \lnext \phi$ \hfill \funax
  \item $\lalways \phi \limplies (\phi \land \lnext\lalways\phi)$\hfill \mixax
  \item $\lalways (\phi \limplies \lnext \phi) \limplies (\phi \limplies \lalways \phi)$ \hfill \indax
  \item $\phi \luntil \psi \limplies \leventually \psi$ \hfill \uoneax
  \item $\phi \luntil \psi \liff \psi \lor (\phi \land \lnext(\phi \luntil \psi))$ \hfill \utwoax
  \setcounter{enumsave}{\theenumi}
 \end{enumerate}
 and rules
 \[
  \lrule{\phi}{\lnext \phi}\,\nextnecrule \, , \qquad\qquad \lrule{\phi}{\lalways\phi}\,\alwaysnecrule \, .
 \]
 
 We use $\LTL$ to denote the Hilbert system given by the axioms and rules for temporal logic above, plus the axioms and rules for propositional logic. 
 
 \subsection{Justification Logic}
 Finally, for the justification logic, we use the counterpart to the multi-agent version of the modal logic $\SFive$, i.e., $\JFive^\numberofagents$ (cf.~\cite{Rub06LC}), with axioms
 \begin{enumerate}
  \setcounter{enumi}{\theenumsave}
  \item $\jbox{t}_\agent (\phi \limplies \psi) \limplies (\jbox{s}_\agent \phi \limplies \jbox{t \tapp s}_\agent \psi)$ \hfill \appax
  \item $\jbox{t}_\agent \phi \lor \jbox{t \tsum s}_\agent \phi$, $\jbox{s}_\agent \phi \limplies  \jbox{t \tsum s}_\agent \phi$ \hfill \sumax
  \item $\jbox{t}_\agent \phi \limplies \phi$ \hfill \refax
  \item $\jbox{t}_\agent \phi \limplies \jbox{\tinspect t}_\agent \jbox{t}_\agent \phi$ \hfill \posintax
  \item $\lnot \jbox{t}_\agent \phi \limplies \jbox{\tneginspect t}_\agent \lnot \jbox{t}_\agent \phi$ \hfill \negintax
  \setcounter{enumsave}{\theenumi}
 \end{enumerate}
 and rule
 \[
   \lrule{\jbox{c}_\agent \phi \in \CS}{\jbox{c}_\agent \phi}\, \constnecrule \, ,
 \]
 where the constant specification $\CS$ is a set of formulae $\jbox{c}_\agent \phi$, where $c \in \CTerms_\agent$ is a proof constant and $\phi$ is an axiom.
 
 We call a constant specification $\CS$ \emph{axiomatically appropriate}, if for every axiom~$\phi$ and agent~$\agent$, there is a constant $c \in \CTerms_\agent$ such that $\jbox{c}_\agent \phi \in \CS$.
 
 For a given constant specification $\CS$, we use $\JFiveLTL_\CS$ to denote the Hilbert system given by the axioms and rules for propositional logic, temporal logic, and justification logic as presented above.
 
 \begin{question}
  What can be done using the following variant of constant necessitation?
  \[
   \lrule{\jbox{c}_\agent \phi \in \CS}{\jbox{c}_\agent \lalways \phi}\
  \]
 \end{question}
 
 \subsection{Additional Principles}
 
 In $\JFiveLTL$, epistemic and temporal properties do not interact. 
 It is therefore a natural question to consider some of the following principles, which create a connection between time and knowledge.
 We assume the language for terms to be augmented in the obvious way.
 
 \begin{enumerate} 
  \item $\jbox{t}_\agent \lalways \phi \limplies \lalways \jbox{\talwaysaccess t}_\agent \phi$ \hfill \alwaysaccessprinciple
  \item $\lalways \jbox{t}_\agent \phi \limplies \jbox{\tgeneralize t}_\agent \lalways  \phi$ \hfill \generalizeprinciple
  \item $\jbox{t}_\agent \lalways \phi \limplies \jbox{\tnextaccess t}_\agent \lnext \phi$ \hfill \nextaccessprinciple
  \item $\jbox{t}_\agent \lnext \phi \limplies \lnext \jbox{\tnext t}_\agent \phi$ \hfill \nextrightshiftprinciple
  \item $\lnext \jbox{t}_\agent \phi \limplies \jbox{\tprev t}_\agent \lnext \phi$ \hfill \nextleftshiftprinciple
 \end{enumerate}
 
 \begin{remark}\mbox{}
  \begin{enumerate}
   \item This is very plausible, if you have evidence that something always is true, then at every point in time you should be able to access this information.
   \item Using evidence this seems more plausible than just using knowledge, as one requires the evidence to be the same at every point in time.
   \item This is similar to $\alwaysaccessprinciple$. One would expect it to be provable from $\alwaysaccessprinciple$ (see below).
   \item This seems plausible: agents do not forget evidence once they have gathered it and can ``take it with them''.
   \item This one seems seems less plausible, as it implies some form of premonition
  \end{enumerate}
 \end{remark}

 When writing $(\textsf{principles}) \proves_\CS \phi$ we mean $\phi$ is provable in $\JFiveLTL_\CS$ with the principles $(\textsf{principles})$ treated as additional axioms, i.e., in particular, this is also reflected by constant necessitation and constant specifications, respectively.  If the constant specification $\CS$ is  clear from the context or not relevant, we omit the corresponding subscript.

\begin{question}
What other principles connecting temporal and epistemic properties might be of interest? See also Section~\ref{sect:SomeConnectingPrinciplesInTemporalModalLogic}.
\end{question}
 
 \section{Some Properties}
 \label{sect:SomeProperties}
 
 In the following, we illustrate the logic by giving two simple (purely temporal) deductions, namely of $\lalways \phi \limplies \lnext \phi$ and $\lalways \phi \limplies \lalways \lalways \phi$, and then show how the connecting principles from the previous section link knowledge and time. We start with these two typical deductions  in $\LTL$.
  
 \begin{lemma}\mbox{}
  \begin{enumerate}
   \item $\proves \lalways \phi \limplies \lnext \phi$
    \item $\proves \lalways \phi \limplies \lalways \lalways \phi$
    \end{enumerate}
 \end{lemma}
 \begin{proof}
  \begin{enumerate}
   \item From $\mixax$ and propositional reasoning we get 
   \begin{align}
\lalways \phi \limplies \phi \label{eq:mixax1}\\
\lalways \phi \limplies \lnext\lalways \phi \label{eq:mixax2}
   \end{align}
   From~\eqref{eq:mixax1} and $\nextnecrule$ we get
   \[
    \lnext(\lalways \phi \limplies \phi)
   \]
   which, in turn, using $\nextkax$ and propositional reasoning gives
   \[
    \lnext\lalways \phi \limplies \lnext \phi \, .
   \]
   By propositional reasoning and using~\eqref{eq:mixax1} we obtain the desired
   \[
    \lalways \phi \limplies \lnext \phi
   \]
   from this. 
   \item The following is a valid instance of $\indax$
   \[
    \lalways (\lalways \phi \limplies \lnext \lalways \phi) \limplies (\lalways \phi \limplies \lalways \lalways \phi) \, .
   \]
   Using~\eqref{eq:mixax2} from the first item, $\alwaysnecrule$, and modus ponens we immediately get the desired result.\qedhere
     \end{enumerate}
 \end{proof}
 
 Now, we show how $\nextaccessprinciple$ can be proved using $\alwaysaccessprinciple$ and $\nextleftshiftprinciple$.

 \begin{lemma}
  For every agent~$\agent$, formula~$\phi$ and term~$t$ there is a term~$s(t)$ such that 
         \[
\alwaysaccessprinciple, \nextleftshiftprinciple \proves \jbox{t}_\agent \lalways \phi \limplies \jbox{s(t)}_\agent \lnext \phi
         \]
 \end{lemma}
 
 \begin{proof}
 Using propositional logic, we can combine the two given principles 
    \begin{align*}
      \jbox{t}_\agent \lalways \phi &\limplies \lalways \jbox{\talwaysaccess}_\agent \phi \, , \\
      \lnext \jbox{\talwaysaccess}_\agent \phi &\limplies \jbox{\tprev \talwaysaccess t}_\agent \lnext \phi \, ,
   \end{align*}
   and the following instance of the principle proved in the first item of the previous lemma
   \begin{align*}
         \lalways \jbox{\talwaysaccess}_\agent \phi &\limplies \lnext \jbox{\talwaysaccess t}_\agent \phi
   \end{align*}
   in order to obtain
   \[
     \jbox{t}_\agent \lalways \phi \limplies \jbox{\tprev\talwaysaccess t}_\agent \lnext \phi \, ,
   \]
   and we are done.
 \end{proof}
 
 \begin{question}
   Can we prove $\nextaccessprinciple$ from $\alwaysaccessprinciple$ without using $\nextleftshiftprinciple$?
 \end{question}
 
 In contrast to the previous deductions, in the following we require our constant specifications to be axiomatically appropriate.

 \begin{lemma}
 Let $\CS$ be an axiomatically appropriate constant specification.
For every agent~$\agent$, formula~$\phi$ and term~$t$ there is a term~$s(t)$ such that 
         \[
\generalizeprinciple \proves_\CS \jbox{t}_\agent \lalways \phi \limplies \jbox{s(t)}_\agent \lalways\lalways \phi
         \]
 \end{lemma}
 \begin{proof}
 From $\constnecrule$ we get
\begin{align*}
&\jbox{c_3} \left( ( \lalways \phi_1 \limplies \phi \land \lnext \lalways \phi) \limplies (\lalways \phi \limplies \lnext \lalways \phi) \right) \, , \\
&\jbox{c_2} ( \lalways \phi_1 \limplies \phi \land \lnext \lalways \phi) \, , \\
&\jbox{c_1} \left( \lalways(\lalways \phi \limplies \lnext \lalways \phi) \limplies (\lalways \phi \limplies \lalways \lalways \phi) \right ) \, ,
\end{align*}
  as these are valid instances of a propositional axiom, $\mixax$, and $\indax$, respectively.
  
  Using $\appax$ and modus ponens, we obtain
  \[
   \jbox{c_2 \tapp c_3} (\lalways \phi \limplies \lnext \lalways \phi) \, .
  \]
  Using $\alwaysnecrule$, this gives
  \[
   \lalways \jbox{c_2 \tapp c_3} (\lalways \phi \limplies \lnext \lalways \phi) \, .
  \]
  From \generalizeprinciple, $\appax$, and modus ponens, we obtain
  \[
   \jbox{\tgeneralize(c_2 \tapp c_3)} \lalways (\lalways \phi \limplies \lnext \lalways \phi) \, .
  \]
  Using $\appax$ and modus ponens two more times, we finally get
  \[
   \jbox{t} \lalways \phi \limplies \jbox{(c_1 \tapp \tgeneralize(c_2 \tapp c_3)) \tapp t} \lalways \lalways \phi \, .
  \]
 \end{proof}
 
 Note that the previous proof  relies on the fact that agents can ``internalize'' deductions.
 This so-called internalization theorem holds in general and is a typical and fundamental property of justification logics.
 
 \begin{theorem}[Internalization]
 Let $\CS$ be an axiomatically appropriate constant specification.  If 
  \[
\generalizeprinciple,\nextaccessprinciple \proves_\CS \phi \, ,   
  \]
 then, for every $1 \leq \agent \leq \numberofagents$ there is a term $t_\agent$ such that 
 \[
\generalizeprinciple, \nextaccessprinciple \proves_\CS \jbox{t}_\agent \phi \, .  
 \]
  
 \end{theorem}
 \begin{proof}
 We proceed by induction on the derivation of $\phi$.
 
 In case $\phi$ is an axiom, the claim is immediate by $\constnecrule$.
 
 In case $\phi$ was derived by modus ponens from $\psi \limplies \phi$ and $\psi$, then, by induction hypothesis, there are term $s_1$ and $s_2$ such that $\jbox{s_1}_\agent (\psi \limplies \phi)$ and $\jbox{s_2} \psi$ are provable.
 Using $\appax$ and modus ponens, we obtain $\jbox{s_1 \tapp s_2}_\agent \phi$.
 
 In case $\phi$ is $\jbox{c}_\agent \psi$, derived using $\constnecrule$, we can use $\posintax$ and modus ponens in order to obtain
 \[
  \jbox{\tinspect c}_\agent \jbox{c}_\agent \psi \, .
 \]
 
 In case $\phi$ is $\lalways \psi$, derived using $\alwaysnecrule$, then, by induction hypothesis, there is a term $s$ such that $\jbox{s}_\agent \psi$ is provable.
 Now, we can use $\alwaysnecrule$ in order to obtain $\lalways \jbox{s}_\agent \psi$ and then $\generalizeprinciple$ and modus ponens to get
 \[
  \jbox{\tgeneralize s}_\agent \lalways \psi \, .
 \]

 Finally, if $\phi$ is $\lnext \psi$, derived using $\nextnecrule$, then, as above, we obtain $\jbox{\tgeneralize s}_\agent \lalways \psi$ and then use $\nextaccessprinciple$ and modus ponens to get
 \[
  \jbox{\tnextaccess \tgeneralize s}_\agent \lnext \psi \, .
 \] 
 \end{proof}
 
 \begin{corollary}
 Let $\CS$ be an axiomatically appropriate constant specification.  If 
  \[
\generalizeprinciple, \alwaysaccessprinciple, \nextleftshiftprinciple \proves_\CS \phi \, ,   
  \]
 then, for every $1 \leq \agent \leq \numberofagents$ there is a term $t_\agent$ such that 
 \[
\generalizeprinciple, \alwaysaccessprinciple, \nextleftshiftprinciple \proves_\CS \jbox{t}_\agent \phi \, .  
 \]
 \end{corollary}
 
 \begin{question}
  Is internalization provable without these additional principles?
 \end{question}

 \section{Semantics}
 \label{sect:Semantics}
 
 Let $\localstates$ be some set of local states.
 A \emph{global state} is a $(\numberofagents+1)$-tuple $\langle l_e, l_1, \dots, l_\numberofagents \rangle \in \localstates^{\numberofagents+1}$.
 A \emph{run}~$r$ is a function from $\N$ to global states, i.e.,  $r: \N \to \localstates^{\numberofagents+1}$.
 Given a run $r$ and $n \in \N$, the global state $(r,n)$ is called \emph{a point}.
 A \emph{system} is a set $\runs$ of runs.

 Let $\CS$ be a constant specification. An \emph{interpreted system}~$\system$ for $\CS$ is a tuple $(\runs, \evidence, \valuation)$ where 
 \begin{itemize}
  \item $\runs$ is a system,
  \item $\evidence_\agent: \Terms_\agent \times \runs \times \N \to \powerset(\Formulae)$ is an $\CS$-admissible evidence function for each $1 \leq \agent \leq \numberofagents$,
  \item $\valuation: \runs \times \N \to \powerset(\Prop)$ is a valuation.
 \end{itemize}
 
 Given two points $(r, n)$ and $(r^\prime, n^\prime)$, we define $(r,n) \accrel_\agent (r^\prime, n^\prime)$ by
 \[
 r(n) =\langle l_e, l_1, \dots, l_\numberofagents \rangle, \, r^\prime(n^\prime) = \langle l^\prime_e, l^\prime_1, \dots, l^\prime_\numberofagents \rangle, \text{ and } l_\agent = l^\prime_\agent \, .  
 \]
 
 A \emph{$\CS$-admissible evidence function} $\evidence_\agent$ is a function satisfying the following conditions.
 For all terms $t,s \in \Terms$ and all points $(r,n)$ and $(r^\prime, n^\prime)$,
 \begin{enumerate}
  \item $\evidence_\agent(r,n,t) \subseteq \evidence_\agent(r^\prime, n^\prime, t)$, whenever $(r,n) \accrel_\agent (r^\prime, n^\prime)$\hfill (montonicity),
  \item if $\jbox{c}_\agent \phi \in \CS$ then $\phi \in \evidence_\agent(r,n,c)$ \hfill (constant specification)
  \item if $\phi \limplies \psi \in \evidence_\agent(r,n,t)$ and $\phi \in \evidence_\agent(r,n,s)$,\\ \phantom{x}\qquad then $\psi \in \evidence_\agent(r,n, t \tapp s)$ \hfill (application)
  \item $\evidence_\agent(r,n,s) \cup \evidence_\agent(r,n,t) \subseteq \evidence_\agent(r,n,s \tsum t)$ \hfill (sum)
  \item if $\phi \in \evidence_\agent(r,n,t)$, then $\jbox{t}_\agent \phi \in \evidence_\agent(r,n,\tinspect t)$ \hfill (positive introspection)
  \item if $\phi \not\in \evidence_\agent(r,n,t)$, then $\neg \jbox{t}_\agent \phi \in \evidence_\agent(r,n,\tneginspect t)$ \hfill (negative introspection)
 \end{enumerate}

 Given an interpreted system $\system=(\runs,\evidence,\valuation)$ for $\CS$, a run $r \in \runs$ and $n \in \N$, we define validity of a formula $\phi$ in $\system$ at point $(r,n)$ inductively by
 \begin{align*}
  (\system, r, n) &\entails P \text{ iff } P \in \valuation(r,n) \, ,\\
  (\system, r, n) &\not\entails \lfalse \, ,\\
  (\system, r, n) &\entails \phi \limplies \psi \text{ iff } (\system, r, n) \not\entails \phi \text{ or } (\system, r, n) \entails \psi \, ,\\
  (\system, r, n) &\entails \lnext \phi \text{ iff } (\system, r, n+1) \entails \phi \, ,\\
  (\system, r, n) &\entails \lalways \phi \text{ iff } (\system, r, n+i) \entails \phi \text{ for all } i\geq 0 \, ,\\
  (\system, r, n) &\entails \phi \luntil \psi \text{ iff there is some } m \geq 0 \text{ such that } (\system, r, n+m) \entails \psi \\ & \qquad\qquad \text{ and } (\system, r, n+i) \entails \phi \text{ for all } 0 \leq i \leq m \, ,\\
  (\system, r, n) &\entails \jbox{t}_\agent \phi \text{ iff }  \phi \in \evidence_\agent(t,r,n)  \text { and } (\system, r^\prime, n^\prime) \entails \phi \\ &\qquad\qquad \text{ for all } r^\prime \in \runs \text{ and } n^\prime \in \N \text{ such that } (r,n) \accrel_\agent (r^\prime,n^\prime) \, .
 \end{align*}
 
 We call an interpreted system \emph{strong} if it has the following additional property
 \begin{itemize}
   \item if $\phi \in \evidence_\agent(r,n,t)$, then $(\system, r, n) \entails \jbox{t}_\agent \phi$ \hfill (strong evidence).
 \end{itemize}
 
 As usual, we write $\system \entails \phi$ if $(\system, r, n) \entails \phi$ for all points $(r,n)$, and we write $\entails_\CS \phi$ if $\system \entails \phi$ for all strong interpreted systems $\system$ for $\CS$.
 
\begin{question}
While in modal epistemic logic, $\SFive$ is typically used, i.e., the accessibility relation is usually an equivalence relation (as $\accrel_\agent$ is here), justification logic is more at home with the justification counterpart to $\SFour$. In order to achieve this, would it be possible to extend interpreted systems with additional, explicit accessibility relations
  \[
   R_\agent \subseteq  \localstates \times \localstates \, ,
  \]
  which are transitive and reflexive? In particular, this would allow dropping the strong evidence requirement.
\end{question}
 
 \section{Soundness}
 \label{sect:Soundness}
 
 \begin{theorem}[Soundness]
  Let $\CS$ be a constant specification.
   \[
   \text{If } \proves_\CS \phi \text{, then }\entails_\CS \phi \, .
   \]
 \end{theorem}
 \begin{proof}
   We proceed by induction on the derivation of $\phi$. Let $\system$ be a system and $(r,n)$ a point.
   
   If $\phi$ is a propositional axiom or derived using modus ponens, the result follows as usual.
   
   In the case of $\nextkax$, assume $(\system, r, n) \entails \lnext (\phi \limplies \psi)$ and $(\system, r, n) \entails \lnext \phi$.
   Then we have $(\system, r, n+1) \entails \phi \limplies \psi$ and $(\system, r, n+1) \entails \phi$.
   Thus, $(\system, r, n+1) \entails \psi$ and we are done.
   
   In the case of $\alwayskax$, assume $(\system, r, n) \entails \lalways (\phi \limplies \psi)$ and $(\system, r, n) \entails \lalways \phi$.
   Then we have $(\system, r, n+i) \entails \phi \limplies \psi$ and $(\system, r, n+i) \entails \phi$ for all $i \geq 0$.
   Thus, $(\system, r, n+i) \entails \psi$ for all $i \geq 0$ and we are done.
   
   In the case of $\funax$, we have 
   \begin{align*}
     (\system, r, n) \entails \lnext \neg \phi\\
     \text{if and only if }& (\system, r, n+1) \entails \neg \phi\\
     \text{if and only if }& (\system, r, n+1) \not\entails \phi\\
     \text{if and only if }& (\system, r, n) \not\entails \lnext\phi\\
     \text{if and only if }& (\system, r, n) \entails \neg \lnext \phi \, .
   \end{align*}
   
   In the case of $\mixax$, assume $(\system, r, n) \entails \lalways \phi$.
   Then we have $(\system, r, n+i) \entails \phi$ for all $i \geq 0$.
   In particular, we have $(\system, r, n) \entails \phi$.
   Furthermore, we also have $(\system, r, n+1+j) \entails \phi$ for all $j \geq 0$.
   Thus, $(\system, r, n+1) \entails \lalways \phi$, which means $(\system, r, n) \entails \lnext \lalways \phi$ and we are done.
   
   In the case of $\indax$, assume $(\system, r, n) \entails \lalways(\phi \limplies \lnext \phi)$ and $(\system, r, n) \entails \phi$.
   Then we have $(\system, r, n+i) \entails \phi \limplies \lnext \phi)$ for all $i \geq 0$.
   By induction on $i$, we can prove $(\system, r, n+i) \entails \phi$ using $(\system, r, n) \entails \phi$ for the induction basis and $(\system, r, n+i) \entails \phi \limplies \lnext \phi)$ for all $i \geq 0$ for the induction step.
   This yields the desired result.
   
   In the case of $\uoneax$, assume $(\system, r, n) \entails \phi \luntil \psi$.
   Thus we have that there is an $m \geq 0$ such that $(\system, r, n+m) \entails \psi$ and $(\system, r, n+i) \entails \phi$ for all $0 \leq i \leq m$.
   In particular, $(\system, r, n+m) \entails \psi$ and thus $(\system, r, n) \entails \leventually \psi$.
   
   In the case of $\utwoax$, for the direction from left to right, assume $(\system, r, n) \entails \phi \luntil \psi$.
   Thus we have that there is an $m \geq 0$ such that $(\system, r, n+m) \entails \psi$ and $(\system, r, n+i) \entails \phi$ for all $0 \leq i \leq m$.
   If $m=0$, we have $(\system, r, n) \entails \psi$ and we are done.
   If $m>0$, we have $(\system, r, n) \entails \phi$, $(\system, r, n+1+(m-1) ) \entails \psi$ and $(\system, r, n+j$ for all $0 \leq j \leq m-1$.
   Thus $(\system, r, n+1) \entails \phi \luntil \psi$ and, in turn, $(\system, r, n) \entails \lnext (\phi \luntil \psi)$.
   
   In the case of $\utwoax$, for the direction from right to left, assume $(\system, r, n) \entails \psi \lor (\phi \land \lnext(\phi \luntil \psi))$.
   If $(\system, r, n) \entails \psi$, the result follows immediately.
   If $(\system,r ,n) \entails \phi \land \lnext(\phi \luntil \psi)$, we have that there is an $m \geq 0$ such that $(\system, r, n+1+m) \entails \psi$ and $(\system, r, n+1+i) \entails \phi$ for all $0 \leq i \leq m$.
   Thus, $(\system, r, n+(m+1)) \entails \psi$ and $(\system, r, n+i) \entails \phi$ for all $0 \leq i \leq (m+1)$ and we are done.
   
   In the case of $\nextnecrule$, by induction hypothesis we have $\entails_\CS \phi$.
   In particular, this means $(\system, r, n+1) \entails \phi$ and we are done.
   
   In the case of $\alwaysnecrule$, by induction hypothesis we have $\entails_\CS \phi$.
   In particular, this means $(\system, r, n+i) \entails \phi$ for all $i \geq 0$ and we are done.
   
   In the case of $\appax$, assume $(\system, r, n) \entails \jbox{t}_\agent(\phi \limplies \psi)$ and $(\system, r, n) \entails  \jbox{s}_\agent \phi$.
   Thus, we  have $\phi \limplies \psi \in \evidence_\agent(r,n,t)$ and $\phi \in \evidence_\agent(r,n,s)$.
   This gives us $\psi \in \evidence_\agent(r,n, t \tapp s)$ and the result follows from the strong evidence condition.
   
   In the first case of $\sumax$, assume $(\system, r, n) \entails \jbox{t}_\agent \phi$.
   Thus, we have $\phi \in \evidence_\agent(r,n,t) \subseteq \evidence_\agent(r,n,t \tsum s)$.
   This gives us $(\system, r,n) \entails \jbox{t \tsum s}_\agent \phi$ by the strong evidence condition.
   The second case follows analogously.
   
   In the case of $\refax$, assume $(\system, r, n) \entails \jbox{t}_\agent \phi$.
   Thus we have $(\system, r^\prime, n^\prime) \entails \phi i$ for all $(r^\prime,n^\prime)$ with $(r,n) \accrel_\agent (r^\prime, n^\prime)$.
   In particular, $(r,n) \accrel (r,n)$, and therefore $(\system, r, n) \entails \phi $  and we are done.
   
   In the case of $\posintax$, assume $(\system, r, n) \entails \jbox{t}_\agent \phi$.
   Thus we have, $\phi \in \evidence_\agent(r,n,t)$.
   From the closure conditions on evidence functions we get $\jbox{t}_\agent \phi \in \evidence_\agent(r,n,\tinspect t)$.
   The strong evidence condition then gives us the desired $(\system, r, n) \entails \jbox{\tinspect t}_\agent \jbox{t}_\agent \phi$.
   
   In the case of $\negintax$, assume $(\system, r, n) \entails \neg \jbox{t}_\agent \phi$.
   by the strong evidence condition, $\phi \not\in \evidence_\agent(r,n,t)$.
   Thus, $\neg \jbox{t}_\agent \phi \in \evidence_\agent( r, n, \tneginspect t)$.
   Now, strong evidence again gives us $(\system, r, n) \entails \jbox{\tneginspect t}_\agent \neg \jbox{t}_\agent \phi$.
   
   Finally, the case of $\constnecrule$ is immediate by the corresponding closure condition on evidence functions and strong evidence.
 \end{proof}
 
 \begin{lemma}\mbox{}
 \begin{enumerate}
 \item $\alwaysaccessprinciple$ is sound for interpreted systems $\system$ satisfying
 \[
 \text{if } \lalways \phi \in \evidence_\agent(r,n,t) \text{, then } \phi \in \evidence_\agent(r,n+k, \talwaysaccess t) \text{ for all } k \geq 0 \, ,
 \]
 for all points $(r,n)$, agents $\agent$, formulae $\phi$ and terms $t$.
 \item $\generalizeprinciple$ is sound for interpreted systems $\system$ satisfying
 \[
  \text{if } \phi \in \evidence_\agent(r, n+k, t) \text{ for all } k \geq 0 \text{, then } \lalways \phi \in \evidence(r, n, \tgeneralize t) \, ,
 \]
for all points $(r,n)$, agents $\agent$, formulae $\phi$ and terms $t$.
 \item $\nextaccessprinciple$ is sound for interpreted systems $\system$ satisfying
 \[
 \text{if } \lalways \phi \in \evidence_\agent(r,n,t) \text{, then } \lnext \phi \in \evidence_\agent(r, n, \tnextaccess t) \, ,
 \]
 for all points $(r,n)$, agents $\agent$, formulae $\phi$ and terms $t$.
 \item $\nextrightshiftprinciple$ is sound for interpreted systems $\system$ satisfying 
 \[
   \text{if } \lnext \phi \in \evidence_\agent(r,n,t) \text{, then } \evidence_\agent(r, n+1, \tnext t) \, ,
 \]
 for all points $(r,n)$, agents $\agent$, formulae $\phi$ and terms $t$.

 \item $\nextleftshiftprinciple$ is sound for interpreted systems $\system$ satisfying
 \[
   \text{if } \phi \in \evidence_\agent(r, n+1, t) \text{, then } \lnext \phi \in \evidence_\agent(r, n, \tprev t)\, ,
 \]
 for all points $(r,n)$, agents $\agent$, formulae $\phi$ and terms $t$.

 \end{enumerate}
 \end{lemma}
 \begin{proof}
 \begin{enumerate}
 \item Assume $(\system, r, n) \entails \jbox{t}_\agent \lalways \phi$.
          Then $\lalways \phi \in \evidence_\agent(r,n,t)$.
          Thus, $\phi \in \evidence(r, n+k, \talwaysaccess t)$ fo all $k \geq 0$.
          By strong evidence, we get $(\system, r, n+k) \entails \jbox{\talwaysaccess t} \phi$ for all $k \geq 0$ and the result follows.
 \item Assume $(\system, r, n) \entails \lalways \jbox{t}_\agent \phi$.
          Then $(\system, r, n+k) \entails \jbox{t}_\agent \phi$ for all $k \geq 0$.
          Thus, $\phi \in \evidence_\agent(r, n+k, t)$ for all $k \geq 0$.
          Therefore, $\lalways \phi \in \evidence_\agent(r,n, \talwaysaccess t)$ and we obtain the result by strong evidence.
 \item Assume $(\system, r, n) \entails \jbox{t}_\agent \lalways \phi$.
          Then $\lalways \phi \in \evidence_\agent(r, n, t)$.
          Thus, $\lnext \phi \in \evidence_\agent(r,n,\tnextaccess t)$ and we obtain the result by strong evidence.
 \item Assume $(\system, r, n) \entails \jbox{t}_\agent \lnext \phi$.
         Then, $\lnext \phi \in \evidence_\agent(r, n, t)$.
         Thus, $\phi \in \evidence_\agent(r, n+1, \tnext t)$ and we obtain the result by strong evidence.
 \item Assume $(\system, r, n) \entails \lnext \jbox{t}_\agent \phi$.
          Then $\phi \in \evidence_\agent(r, n+1, t)$.
          Thus, $\lnext \phi \in \evidence_\agent(r, n, \tprev t)$ and we obtain the result by strong evidence.
 \end{enumerate}
 \end{proof}
 
 \begin{question}
  Do models satisfying these additional conditions exist at all?
 \end{question}
 
 \begin{question}
 Are there any (less obvious) semantic conditions guaranteeing soundness for these principles?
 \end{question}

 \begin{question}
 How can one show completeness? Adapting the proof from~\cite{HvdMV04}, where interpreted systems are obtained from a (finite) canonical model construction might be a feasible route, but the presence of $\lalways$ might make it more cumbersome. Using infinite canonical models might require some form of model surgery, e.g., filtrations, as we are dealing with fixed points.
 \end{question}
   
 \section{Conclusion}
 \label{sect:Conclusion}
 
 We have sketched an axiomatization for a justification logic of knowledge and time, discussed connecting principles between knowledge and time, illustrated the logic with sample derivations, and shown the internalization theorem and soundness. 
 In the course of the presentation, we have raised questions indicating further directions of work. 
 Most prominently, completeness proofs are currently missing.
 
 Besides the questions posed above, there are various further routes of research such a logic might open.
 We will outline these questions in the following, in no particular order.
 
 \begin{question}
 Can one build Mkrtychev-style~\cite{Mkr97LFCS} interpreted systems? These would be models which do not require the accessibility relation $\accrel_\agent$, but solely rely on the evidence function for determining knowledge.
 \end{question}
 
 \begin{question}
 How can the typical examples, e.g., protocols related to message transmission, be formalized in the logic presented above? See, e.g.,~\cite{HZ92}. For example, one might consider principles such as
 \begin{itemize}
   \item $\jbox{t}_\agent \phi \limplies \lnext \jbox{\textsf{sent}^\agent_j(t)}_j \phi$, or
  \item $\jbox{t}_\agent \phi \limplies \leventually \jbox{\textsf{sent}^\agent_j(t)}_j \phi$.
 \end{itemize}
 \end{question}
 
 \begin{question}
 What happens if we require operations on justification terms to take time, e.g.,
 \[
   \jbox{t}_\agent \phi \limplies \lnext \jbox{\tinspect t}_\agent \jbox{t}_\agent \phi \, ?
 \]
 This might also relate to the logical omniscience problem~\cite{ArtKuz14APAL}.
 \end{question}
 
 \begin{question}
 What does a justification logic for knowledge and \emph{branching} time look like? See also~\cite{vdMW03}.
 \end{question}
 
 \begin{question}
 Can dynamic justification logic be translated into temporal justification logic akin to~\cite{vDvdHR13}? See also~\cite{BucKuzStu14}.
 \end{question}
 
 Finally, one might also wonder whether the temporal modalities themselves can be justified.
 This question is independent of the presentation above, more information can be found in Section~\ref{sect:AJustifiedTemporalLogic}.
 
 \begin{question}
 What would a justified temporal logic look like?
 \end{question}
 
 \bibliographystyle{alpha}
 \bibliography{bibliography,JLBibliography}
 
 \newpage
 
 \appendix
 
 \section{An Alternative Presentation of Temporal Logic}
 \label{sect:AnAlternativePresentationOfTemporalLogic}
 
 Linear temporal logic is often presented using an induction rule as follows, see, e.g.,~\cite{HvdMV04}, with axioms
 \begin{itemize}
  \item $\lnext( \phi \limplies \psi) \limplies (\lnext \phi \limplies \lnext \psi)$ \hfill \nextkax
  \item $\lnext \lneg \phi \liff \lneg \lnext \phi$ \hfill \funax
  \item $\phi \luntil \psi \liff \psi \lor (\phi \land \lnext(\phi \luntil \psi))$ \hfill \utwoax
 \end{itemize}
 and rules
 \[
  \lrule{\phi}{\lnext \phi}\,\nextnecrule \, ,\qquad\qquad \lrule{\chi \limplies \lneg \psi \land \lnext \chi}{\chi \limplies \lneg(\phi \luntil \psi)}\,\uindrule \, .
 \]
 Here, $\leventually$ and $\lalways$ are defined by
 \begin{align*}
  \leventually \phi &\colonequals \ltrue \luntil \phi \, ,\\
  \lalways \phi &\colonequals \lneg \leventually \lneg \phi \, .
 \end{align*}
 
 We use $\LTLalt$ to denote the Hilbert system given by the alternative axioms and rules for temporal logic above, plus the axioms and rules for propositional logic.
 
 \subsection{Relationship between Temporal Logic and Alternative Presentation of Temporal Logic}
 
 We will now show that this presentation is equivalent to the one previously given. We start by showing some auxiliary results.
 
 \begin{lemma}
  The following rules are derivable in $\LTL$:
  \begin{enumerate}
   \item \label{lem:temptoalt:nexttoalways1} $\lrule{\phi \limplies \lnext \phi}{\phi \limplies \lalways \phi}$,
   \item \label{lem:temptoalt:nexttoalways2} $\lrule{\chi \limplies \phi \land \lnext \chi}{\chi \limplies \lalways \phi}$,
   \item \label{lem:temptoalt:alwaystountil} $\lrule{\chi \limplies \lalways \lneg \psi}{\chi \limplies \lneg ( \phi \luntil \psi )}$.
  \end{enumerate}
 \end{lemma}
 \begin{proof}
 \begin{enumerate}
  \item Assume $\phi \limplies \lnext\phi$. 
        By $\alwaysnecrule$ we obtain $\lalways(\phi \limplies \lnext \phi)$.
        Using $\indax$ and $\mprule$, we get the desired $\phi \limplies \lalways \phi$.
  \item Assume $\chi \limplies \phi \land \lnext \chi$.
        By propositional reasoning, we get both
        \begin{align}
         \chi &\limplies \lnext \chi \label{eq:chiimpliesnextchi}\\
         \chi &\limplies \phi \label{eq:chiimpliesphi}
        \end{align}
        From~\eqref{eq:chiimpliesnextchi} we get 
        \begin{equation}
         \chi \limplies \lalways \chi \label{eq:chiimpliesalwayschi}
        \end{equation}
        by using Lemma~\ref{lem:temptoalt}, item~\ref{lem:temptoalt:nexttoalways1} above.
        
        From~\eqref{eq:chiimpliesphi} we get $\lalways(\chi \limplies \phi)$ by $\alwaysnecrule$ and from this we obtain
        \begin{equation}
         \lalways \chi \limplies \lalways \phi \label{eq:alwayschiimpliesalwaysphi}
        \end{equation}
        by using $\alwayskax$ and $\mprule$.
        
        Using propositional reasoning we can combine \eqref{eq:chiimpliesalwayschi}~and~\eqref{eq:alwayschiimpliesalwaysphi} in order to obtain the desired
        \[
         \chi \limplies \lalways \phi \, .
        \]
  \item Immediate by using propositional reasoning and the contrapositive of $\uoneax$ which is $\lalways \lneg \psi \limplies \lneg (\phi \luntil \psi)$.\qedhere
 \end{enumerate}
 \end{proof}
 
 Combining these results, we obtain the following:
 
 \begin{lemma}
   \label{lem:temptoalt}
  The rule~$\uindrule$ is derivable in $\LTL$.
 \end{lemma}
 \begin{proof}
  Assume $\chi \limplies \lneg \psi \land \lnext \chi$.
  By Lemma \ref{lem:temptoalt}, item~\ref{lem:temptoalt:nexttoalways2} we get $\chi \limplies \lalways \lneg \psi$.
  Using Lemma \ref{lem:temptoalt}, item~\ref{lem:temptoalt:alwaystountil} we get $\chi \limplies \lneg(\phi \luntil \psi)$ and we are done.
 \end{proof}

For the other direction, we have:

 \begin{lemma}
  \label{lem:alttotemp}
  The following axioms and rules are derivable in $\LTLalt$:
  \begin{enumerate}
   \item \label{lem:alttotemp:mixax} $\mixax$
   \item $\uoneax$
   \item $\alwaysnecrule$
   \item $\alwayskax$
   \item $\indax$
  \end{enumerate}
 \end{lemma}
 \begin{proof}
  Remember that $\lalways$ is a defined connective in $\LTLalt$, i.e., $\lalways \phi \colonequals \lneg( \ltrue \luntil \lneg \phi)$.
  \begin{enumerate}
    \item The following is an instance of $\utwoax$
          \[
           \lneg \phi \lor (\ltrue \land \lnext( \ltrue \luntil \lneg \phi)) \limplies \ltrue \luntil \lneg \phi \, .
          \]
          Using propositional reasoning, this is equivalent to
          \[
           \lneg \phi \lor \lnext( \ltrue \luntil \lneg \phi)) \limplies \ltrue \luntil \lneg \phi \, .
          \]
          Taking the contrapositive of this and using propositional reasoning and $\funax$, we obtain
          \[
           \lneg(\ltrue \luntil \lneg \phi) \limplies (\phi \land \lnext \lneg (\ltrue \luntil \lneg \phi)) \, ,
          \]
          which is the desired result
          \[
           \lalways \phi \limplies (\phi \land \lnext \lalways \phi \, .
          \]
    \item The following is an instance of $\uoneax$
          \[
           \psi \lor (\ltrue \land \lnext(\ltrue \luntil \psi)) \limplies \ltrue \luntil \psi \, .
          \]
          Taking the contrapositive, using propositional reasoning and $\funax$, we obtain
          \[
           \lneg(\ltrue \luntil \psi) \limplies \lneg \psi \land \lnext \lneg(\ltrue \luntil \psi) \, .
          \]
          Now we can use $\uindrule$ in order to obtain
          \[
           \lneg(\ltrue \luntil \psi) \limplies \lneg (\phi \luntil \psi) \, ,
          \]
          whose contrapositive
          \[
           \phi \luntil \psi \limplies \ltrue \luntil \psi
          \]
          is the desired
          \[
           \phi \luntil \psi \limplies \leventually \psi \, .
          \]
    \item Assume $\phi$. By propositional reasoning and $\nextnecrule$ we get
          \[
           \ltrue \limplies \lneg \lneg \phi \land \lnext \ltrue \, .
          \]
          Using $\uindrule$, we get
          \[
           \ltrue \limplies \lneg(\ltrue \luntil \lneg \phi ) \, ,
          \]
          which is equivalent to
          \[
           \lneg(\ltrue \luntil \lneg \phi) \, ,
          \]
          which is the desired
          \[
           \lalways \phi \, .
          \]
    \item By Lemma~\ref{lem:alttotemp}, item~\ref{lem:alttotemp:mixax}, the following instances of $\mixax$ are provable
          \begin{align*}
           \lalways \phi &\limplies (\phi \land \lnext \lalways \phi) \, ,\\
           \lalways (\phi \limplies \psi) &\limplies ( (\phi \limplies \psi) \land \lnext\lalways(\phi \limplies \psi) \, .
          \end{align*}
          Using propositional reasoning to combine these, we obtain
          \[
           (\lalways \phi \land \lalways (\phi \limplies \psi)) \limplies (\phi \land (\phi \limplies \psi) \land \lnext \lalways \phi \land \lnext\lalways(\phi \limplies \psi)) \, .
          \]
          Using propositional reasoning and $\nextkax$, we get
          \[
           (\lalways \phi \land \lalways (\phi \limplies \psi)) \limplies (\psi \land \lnext (\lalways \phi \land \lalways(\phi \limplies \psi)) \, .
          \]
          Note that this has the form
          \[
           \chi \limplies (\psi \land \lnext \chi )\, ,
          \]
          where $\chi = \lalways \phi \land \lalways (\phi \limplies \psi)$.
          Now we can use $\uindrule$ in order to obtain
          \[
           \chi \limplies \lneg (\ltrue \luntil \lneg \psi) \, ,
          \]
          which is
          \[
           (\lalways \phi \land \lalways(\phi \limplies \psi)) \limplies \lalways \psi \, ,
          \]
          which in turn is propositionally equivalent to the desired $\alwayskax$.
    \item Using Lemma~\ref{lem:alttotemp}, item~\ref{lem:alttotemp:mixax} as above and propositional reasoning we have
          \[
           \phi \land \lalways (\phi \limplies \lnext \phi) \limplies ( \phi \land (\phi \limplies \lnext \phi) \land \lnext\lalways(\phi \limplies \lnext \phi))\, ,
          \]
          which, by using further propositional reasoning and $\nextkax$ can be turned into
          \[
           \phi \land \lalways (\phi \limplies \lnext \phi) \limplies ( \phi \land \lnext (\phi \land \lalways(\phi \limplies \lnext \phi))\, .
          \]
          This has the form
          \[
           \chi \limplies \phi \land \lnext \chi \, ,
          \]
          where $\chi = \phi \land \lalways (\phi \limplies \lnext \phi)$.
          Hence we can use $\uindrule$ in order to obtain
          \[
           \chi \limplies \lneg (\ltrue \luntil \lneg \phi) \, ,
          \]
          which is 
          \[
           \phi \land \lalways(\phi \limplies \lnext \phi) \limplies \lalways \phi \, ,
          \]
          which is trivially equivalent to the desired $\indax$.
          \qedhere
  \end{enumerate}
 \end{proof}

Finally, putting everything together, we obtain the desired equivalence.

 \begin{theorem}
  $\LTL \proves \phi$ if and only if $\LTLalt \proves \phi$.
 \end{theorem}
 \begin{proof}
  Immediate by induction on the derivation and using Lemma~\ref{lem:alttotemp} and Lemma~\ref{lem:temptoalt}.
 \end{proof}
 
 \newpage
 
 \section{Some Connecting Principles in Temporal Modal Logic}
 \label{sect:SomeConnectingPrinciplesInTemporalModalLogic}
 
 The following definitions and facts are taken from~\cite{HvdMV04}. Here we adapt the language and use $\lknows_\agent \phi$ for the modality ``agent~$\agent$ knows $\phi$'' and its dual $\lconsiders_\agent \phi \colonequals \lneg \lknows_\agent \lneg \phi$. We first need the following preliminary definitions.

 \begin{definition}
 For agent~$\agent$ and point $(r,n)$, we define
 \begin{enumerate}
  \item the \emph{local state sequence} $\localstatesequence_\agent(r,n)$ to be the sequence of local states of agent~$\agent$ in run $r$ up to and including time~$n$, but with consecutive repetitions omitted,
  \item the \emph{future local state sequence} $\futurelocalstatesequence_\agent(r,n)$ to be the sequence of local states of agent~$\agent$ in run $r$ starting from time~$n$, but with consecutive repetitions omitted.
  \end{enumerate}
 \end{definition}
 Using these definitions, we can define the following notions.
 \begin{definition}\mbox{}
  \begin{enumerate}
   \item A system has a \emph{unique intial state} if for all runs $r,r^\prime \in \runs$
         \[
          r(0) = r^\prime(0) \, .
         \]

   \item A system is \emph{synchronous} if for all agents $\agent$, points $(r,n)$, and $(r^\prime, n^\prime)$
         \[
          (r,n) \accrel_\agent (r^\prime,n^\prime) \text{ implies } n = n^\prime \, ,
         \]
         i.e., agents do know the (global) time.
   \item Agent~$\agent$ \emph{has perfect recall}  in a system, if for all points $(r,n)$ and $(r^\prime, n^\prime)$
            \[
            (r,n) \accrel_\agent (r^\prime, n^\prime) \text{ implies }\localstatesequence_\agent(r,n) =  \localstatesequence_\agent(r^\prime, n^\prime)\, ,
            \]
            i.e., if the agent considers $r^\prime$ possible, it must have considered it possible at all points in the past.
   \item Agent~$\agent$ \emph{does not learn}  in a system, if for all points $(r,n)$ and $(r^\prime, n^\prime)$
            \[
            (r,n) \accrel_\agent (r^\prime, n^\prime) \text{ implies }\futurelocalstatesequence_\agent(r,n) =  \futurelocalstatesequence_\agent(r^\prime, n^\prime)\, ,
            \]
            i.e., if the agent considers $r^\prime$ possible, it will do so at all points in the future.
     \end{enumerate}
 \end{definition}
 Corresponding to these semantic notions, we have the following principles
 \begin{definition}\mbox{}
 \begin{itemize}
  \item[(KT3)] $\lknows_\agent \phi \land \lnext( \lknows_\agent \psi \land \lneg \lknows_\agent \chi)\limplies \lconsiders_\agent ( (\lknows_\agent \phi) \luntil ((\lknows_\agent \psi) \luntil \lneg \chi ) ) $ \hfill \prax
  \item[(KT1)] $\lknows_\agent \lalways \phi \limplies \lalways \lknows_\agent \phi$ \hfill \notquiteprax
  \item[(KT2)] $\lknows_\agent \lnext \phi \limplies \lnext \lknows_\agent \phi$ \hfill \prsyncax
  \item[(KT4)] $\lknows_\agent \phi \luntil \lknows_\agent \psi \limplies \lknows_\agent(\lknows_\agent  \phi \luntil \lknows_\agent \psi)$ \hfill \nlax
  \item[(KT5)] $\lnext \lknows_\agent \phi \limplies \lknows_\agent \lnext \phi$ \hfill \nlsyncax
  \item[] $\lknows_\agent \phi \liff \lknows_1 \phi$ \hfill \knowexchax
 \end{itemize}
 \end{definition}
 Finally, the following relationships hold between these principles and semantic notions:
 \begin{itemize}
  \item $\prax$ (strictly) implies $\notquiteprax$,
  \item $\prax$ gives a sound and complete axiomatization for systems with \emph{perfect recall} (with or without \emph{unique initial state}),
  \item $\prsyncax$ gives a sound and complete axiomatization for \emph{synchronous} systems with \emph{perfect recall} (with or without \emph{unique initial state}),
  \item $\nlax$ gives a sound and complete axiomatization for systems with \emph{no learning} (without \emph{unique initial state}),
  \item $\nlsyncax$ gives a sound and complete axiomatization for \emph{synchronous} systems with \emph{no learning} without \emph{unique initial state},
  \item $\prax$ and $\nlax$ give a sound and complete axiomatization for systems with \emph{perfect recall} and \emph{no learning}  \emph{without unique initial state},
  \item $\prax$ and $\nlax$ give a sound and complete axiomatization for \emph{single-agent} (i.e., $\numberofagents = 1$) systems with \emph{perfect recall} and \emph{no learning}  \emph{with unique initial state}),
  \item $\prsyncax$ and $\nlsyncax$ give a sound and complete axiomatization for \emph{synchronous} systems with \emph{perfect recall} and \emph{no learning} \emph{without unique initial state},
  \item $\prsyncax$ and $\nlsyncax$ and $\knowexchax$ give a sound and complete axiomatization for \emph{synchronous} systems with \emph{no learning} and \emph{with unique initial state} (with or without \emph{perfect recall}),
  \item Systems with \emph{no learning} and \emph{with unique initial state} with {more than one agent} (i.e., $\numberofagents \geq 2$) do not have a recursive axiomatic characterization since the validity problem is co-r.e.-complete.
 \end{itemize}

 \newpage
 
 \section{A Justified Temporal Logic?}
 \label{sect:AJustifiedTemporalLogic}
 
 Looking at $\LTL$, one might also be tempted to create a justified temporal logic, where we introduce justifications for temporal modalities. 
 This might look as follows.
 
 \subsection{Syntax}
 
  We have three different types of terms, corresponding to the three different temporal modalities:
 \begin{align*}
  t^\lnext &\coloncolonequals c^\lnext \mid x^\lnext \mid t^\lnext \tapp t^\lnext \mid \tahead(t^\lalways) \mid \tuhead(t^\luntil) \, ,\\
  t^\lalways &\coloncolonequals c^\lalways \mid x^\lalways \mid \taind(t^\lalways,t^\lnext) \mid \tatail(t^\lalways)  \, ,\\
  t^\luntil &\coloncolonequals c^\luntil \mid x^\luntil \mid \tuappend(t^\lnext, t^\luntil) \mid \tutail(t^\luntil) \, .
 \end{align*}
 
 Formulae are as usual, but with modalities replaced by terms:
 \[
  \phi \coloncolonequals p \mid \lfalse \mid \phi \limplies \phi \mid \jbox{t^\lnext}_\lnext \phi \mid \jbox{t^\lalways}_\lalways \phi \mid \jbox{t^\luntil}_\luntil(\phi,\phi) \, .
 \]
 
 \subsection{Axiomatisation}
 
 \begin{enumerate}
  \setcounter{enumi}{-1}
  \item all propositional tautologies
  \item $\jbox{t}_\lnext( \phi \limplies \psi) \limplies (\jbox{s}_\lnext \phi \limplies \jbox{t \tapp s}_\lnext \psi)$ \hfill \nextappax
  \item $\jbox{t}_\lalways( \phi \limplies \psi) \limplies (\jbox{s}_\lalways \phi \limplies \jbox{t \tapp s}_\lalways \psi)$ \hfill \alwaysappax
  \item $\jbox{t}_\lnext \lneg \phi \liff \lneg \jbox{t}_\lnext \phi$ \hfill \funax
  \item $\jbox{t}_\lalways \phi \limplies (\phi \land \jbox{\tahead(t)}_\lnext\jbox{\tatail(t)}_\lalways\phi)$\hfill \amixax
  \item $\jbox{t}_\lalways (\phi \limplies \jbox{s}_\lnext \phi) \limplies (\phi \limplies \jbox{\taind(t,s)}_\lalways \phi)$ \hfill \aindax
  \item $\jbox{t,s}_\luntil(\phi,\psi) \limplies \neg \jbox{s}_\lalways \neg \psi$ \hfill \uax
  \item $\psi \lor (\phi \land \jbox{t_1}_\lnext(\jbox{t_2}_\luntil(\phi, \psi))) \limplies \jbox{\tuappend(t_1,t_2),s}_\luntil(\phi)$ \hfill \uindax 
  \item $\jbox{t}_\luntil(\phi, \psi) \limplies \psi \lor (\phi \land \jbox{\tuhead(t)}_\lnext(\jbox{\tutail(t)}_\luntil(\phi, \psi)))$ \hfill \umixax
 \end{enumerate}
 and rules
 \[
  \lrule{\phi \quad \phi \limplies \psi}{\psi}\,\mprule \, , \qquad\qquad \lrule{\jbox{c}_\lnext \phi \in \CS }{\jbox{c}_\lnext \phi}\,\nextnecrule \, , \qquad\qquad \lrule{\jbox{c}_\lalways \phi \in \CS}{\jbox{c}_\lalways \phi}\,\alwaysnecrule \, .
 \]
 
 Roughly speaking, one could consider justifications for $\lalways$ and $\luntil$ as lists of terms, that are either generated by induction (in the case of $\lalways$) or by appending elements to a given list (in the case of $\luntil$).
 
 However, this is just a preliminary sketch and it is not clear, what the actual semantics of such a system would be.
 Whether and what sense such a system makes has to be left for future work.

\end{document}